\documentclass{article}
\usepackage{enumitem}
\usepackage{amsmath}
\usepackage{amsthm}
\usepackage{xcolor}
\usepackage[numbers]{natbib}
\usepackage{graphicx}
\usepackage{booktabs}
\usepackage{tabularx}
\usepackage{arydshln}
\DeclareMathOperator{\Tr}{Tr}
\DeclareMathOperator{\Var}{Var}
\DeclareMathOperator*{\argmax}{arg\,max}
\newcommand{\one}{\mathbf{1}}
\newcommand{\Xdeg}{\mbox{$X$-degree}}
\newcommand{\Xnb}{\mbox{$X$-{NB}}}
\theoremstyle{remark}
\newtheorem{remark}{Remark}
\newtheorem{problem}{Problem}
\newtheorem{corollary}{Corollary}
\newtheorem{theorem}{Theorem}[section]
\newtheorem{lemma}[theorem]{Lemma}
\newtheorem{proposition}[theorem]{Proposition}

\usepackage[linesnumbered,boxed]{algorithm2e}

\begin{document}

\title{Node Immunization with Non-backtracking Eigenvalues}

\author{
  Leo Torres \\
  \texttt{leo@leotrs.com} \\
  Network Science Institute, \\
  Northeastern University \\
  \and
  Kevin S. Chan \\
  \texttt{kevin.s.chan.civ@mail.mil} \\
  U.S. Army Research Lab
  \and
  Hanghang Tong \\
  \texttt{htong@illinois.edu} \\
  Department of Computer Science, \\ University of Illinois at Urbana-Champaign
  \and
  Tina Eliassi-Rad \\
  \texttt{tina@eliassi.org} \\
  Network Science Institute and \\
  Khoury College of Computer Sciences, \\ Northeastern University
}

\date{}
\maketitle

\begin{abstract}
The \emph{non-backtracking matrix} and its eigenvalues have many applications in network science and graph mining, such as node and edge centrality, community detection, length spectrum theory, graph distance, and epidemic and percolation thresholds. Moreover, in network epidemiology, the reciprocal of the largest eigenvalue of the non-backtracking matrix is a good approximation for the epidemic threshold of certain network dynamics. In this work, we develop techniques that identify which nodes have the largest impact on the leading non-backtracking eigenvalue. We do so by studying the behavior of the spectrum of the non-backtracking matrix after a node is removed from the graph. From this analysis we derive two new centrality measures: \emph{\Xdeg{}} and \emph{\mbox{X-non-backtracking centrality}}. We perform extensive experimentation with targeted immunization strategies derived from these two centrality measures. Our spectral analysis and centrality measures can be broadly applied, and will be of interest to both theorists and practitioners alike.
\end{abstract}


\section{Introduction}\label{sec:introduction}
A \emph{non-backtracking walk} in a graph is a sequence of pairwise adjacent edges such that no edge is traversed twice in succession, i.e., the walk does not contain \emph{backtracks}. Non-backtracking walks are known to mix faster than standard random walks \citep{alon07}, whereas \emph{non-backtracking cycles} (i.e. closed non-backtracking walks) contain important topological information from the so-called \emph{length spectrum} of the graph \citep{Torres2019}. The associated \emph{non-backtracking matrix} is the unnormalized transition matrix of a random walker that does not trace backtracks, and it has many applications such as community detection \citep{bordenave2015,krzakala2013spectral}, influencer identification \citep{morone2016collective,morone2015influence}, graph distance \citep{Torres2019,mellor2019}, etc. Additionally, the \emph{non-backtracking centrality} of nodes, defined in terms of the principal eigenvector of the non-backtracking matrix, has more desirable properties than the standard eigenvector centrality \citep{martin2014localization}. 
The non-backtracking framework has also been adapted to directed networks \citep{arrigo2017directed}, weighted networks \citep{kempton2016non}, and multi-layer networks \citep{arrigo2018exponential}. In this paper, to avoid repetition we use the prefix ``NB'' to mean non-backtracking. For example, we refer to the non-backtracking matrix as the NB-matrix.

The eigenvalues of the NB-matrix (or NB-eigenvalues for short) are related to certain kinds of spreading dynamics. \citet{karrer2014percolation} and \citet{hamilton2014tight} showed that the percolation threshold is approximated by the inverse of the largest NB-eigenvalue $\lambda_1$. This implies that the epidemic threshold of susceptible-infectious-recovered (SIR) dynamics can also be approximated by $\lambda_1$ \citep{pastor2015epidemic,newman2002spread}. Furthermore, \citet{shrestha2015message} argued the same for susceptible-infectious-susceptible (SIS) dynamics, though \citet{castellano2018relevance} highlight that this may only hold for networks with certain amounts of degree heterogeneity, and propose a fully non-backtracking version of SIS dynamics where the NB-walks also play a large role. Whether one is talking about the percolation threshold or the epidemic threshold on SIR or SIS dynamics, $\lambda_1$ of the NB-matrix provides a better approximation to the true epidemic threshold than the largest eigenvalue of the adjacency matrix \citep{shrestha2015message,karrer2014percolation}.

Given the importance of the largest NB-eigenvalue in network dynamics, we ask: \textbf{which nodes have the largest influence on the largest NB-eigenvalue?} In the cases of SIR and SIS dynamics, answering this question will lead to targeted immunization strategies, as it is equivalent to asking which are the nodes whose removal from the network causes the epidemic threshold to increase the most. In the case of percolation, this is equivalent to determining which nodes' removal will put the network closer to splitting into many connected components. Operationally, we frame this question as follows. 

\begin{problem}
Consider a graph $G$ with largest NB-eigenvalue $\lambda_1$. Given an arbitrary node $c$, define $\lambda_1(c)$ as the largest NB-eigenvalue of the network after removing $c$. Define $\lambda_1 - \lambda_1(c)$ as the \emph{eigen-drop induced by $c$}. \textbf{Which node $c$ induces the maximum eigen-drop?}
\end{problem}

The contributions of the present work are as follows:
\begin{itemize}
\item We develop the spectral theory of the NB-matrix to study the behavior of its eigenvalues under the removal of a node.
\item For Problem 1, we propose two new centrality measures, \emph{\Xdeg{}} and \emph{\mbox{$X$-non-backtracking} (or \Xnb{}) centrality}. Further, \Xdeg{} can be computed in approximately log-linear time in the number of nodes.
\item Our experiments show that immunization strategies induced by \Xdeg{} and \Xnb{} centrality are more effective than other methods.
\end{itemize}

In Section~\ref{sec:background} we present the necessary background theory. In Section~\ref{sec:theory} we develop a spectral perturbation theory of the NB-matrix. We use this theory in Section~\ref{sec:immunization} to introduce two new centrality measures and argue why they are effective at identifying nodes with largest eigen-drops. In Section~\ref{sec:related-work} we review previous studies related to the present work. In Section~\ref{sec:experiments} we provide experimental evidence for our claims. We conclude the paper in Section~\ref{sec:conclusion}.


\section{Background}\label{sec:background}

Let $G$ be a simple undirected graph with node set $V$ and edge set $E$. We consider the set of directed edges $\overline{E}$ where each undirected edge $(i, j) \in E$ gives rise to two directed edges $i \to j \in \overline{E}$ and $j \to i \in \overline{E}$. A \emph{walk} in $G$ is a sequence of directed edges $i_1 \to j_1$, \ldots, $i_k \to j_k$, where $j_r = i_{r+1}$ for each $r=1,\ldots, k-1$. Here, $k$ is the \emph{length} of the walk. A walk is \emph{closed} if $j_k = i_1$. A \emph{backtrack} is a walk of length $2$ of the form $i \to j, j \to i$. A walk is a \emph{non-backtracking walk}, or \emph{NB-walk}, if no two consecutive edges in it form a backtrack. The \emph{non-backtracking matrix}, or \emph{NB-matrix}, $B$ is the unnormalized transition matrix of a walker that does not perform backtracks. Concretely, $B$ is indexed in the rows and columns by elements of $\overline{E}$. Let $m = |E|$, then $B$ is of size $2m \times 2m$, and it is defined by
\begin{equation}\label{eqn:nbm-element}
B_{k \to l, i \to j} = \delta_{jk} \left( 1 - \delta_{il} \right),
\end{equation}
where $\delta$ is the Kronecker delta. In words, $B_{k \to l, i \to j}$ is $1$ iff $j = k$ and $i \to j, j \to l$ is not a backtrack. Notably, the powers of $B$ count the number of NB-walks in $G$, i.e. $\left( B^r \right)_{k \to l, i \to j}$ is the number of NB-walks that start with $i \to j$ and end with $k \to l$ of length $r+1$.

Among other applications, the NB-matrix has been used to define a notion of node centrality that has more desirable properties than the usual eigenvector centrality \citep{martin2014localization,radicchi2016leveraging}. Concretely, let $\lambda$ be the largest eigenvalue of $B$ and let $\mathbf{v}$ be the corresponding unit right eigenvector. By Perron-Frobenius theory, $\lambda$ is positive, real, and has multiplicity one, while $\mathbf{v}$ can be chosen to be non-negative. The \emph{non-backtracking centrality} of a node $i$ is defined as
\begin{equation}\label{eqn:nb-centrality}
\mathbf{v}^{i} = \sum_j a_{ij} \mathbf{v}_{j \to i},
\end{equation}
where $A = \left( a_{ij} \right)$ is the adjacency matrix of $G$. Now let $D$ be the diagonal matrix with the degree of each node, and let $\mathbf{v}_{aux}$ be the left principal eigenvector of
\begin{equation}\label{eqn:aux-nb-matrix}
B_{aux} =
\left(
\begin{array}{cc}
0 & D - I_n \\
- I_n & A
\end{array}
\right),
\end{equation}
where $I_r$ is the identity matrix of size $r$. We have that $\mathbf{v}_{aux} = \left( \mathbf{f}, -\lambda \mathbf{f} \right)$, where $\mathbf{f}$ is of size $n$, and it is known that $\mathbf{f}$ is parallel to the vector of NB-centralities, $\mathbf{f}^i \propto \mathbf{v}^i$ \cite{martin2014localization}. It is more efficient to use $B_{aux}$ than $B$ when computing the NB-centrality, since the former matrix is of size $2n \times 2n$, where $n = |V|$.

The NB-matrix is not symmetric and therefore its eigenvalues, other than the largest one, can be complex numbers. Even so, it contains a subtle structure, sometimes called PT-symmetry \citep{bordenave2015}. Indeed, let $P$ be the matrix such that $P \mathbf{x}_{i \to j} = \mathbf{x}_{j \to i}$ for any vector $\mathbf{x}$ indexed by $\overline{E}$. It is readily checked that (i) the product $PB$ is symmetric, and (ii) there exists a basis where $P$ can be written as
\begin{equation}\label{eqn:p}
P =
\left(
\begin{array}{cc}
0 & I_m \\
I_m & 0
\end{array}
\right).
\end{equation}


\section{Non-backtracking eigenvalues under node removal}\label{sec:theory}

We are interested in the behavior of the NB-eigenvalues when we remove a node from $G$. Suppose the target node we want to remove is $c \in V$, and partition the edges in $\overline{E}$ as those that are incident to $c$ and those that are not. Sort the rows and columns of $B$ accordingly, so that it takes the block form
\begin{equation}\label{eqn:block-form}
B =
\left(
\begin{array}{cc}
B' & D \\
E & F
\end{array}
\right),
\end{equation}
as shown in Figure~\ref{fig:blue-yellow}. Here, $B'$ is the NB-matrix of the graph after node $c$ is removed, while $F$ is the NB-matrix of the star graph centered at $c$; if $d$ is the degree of $c$, then $F$ is of size $2d \times 2d$. Further, $D$ is indexed in the rows by directed edges not incident to $c$, and in the columns by directed edges incident to $c$, and vice versa for $E$.

\begin{figure}
    \centering
    \includegraphics[width=0.8\textwidth]{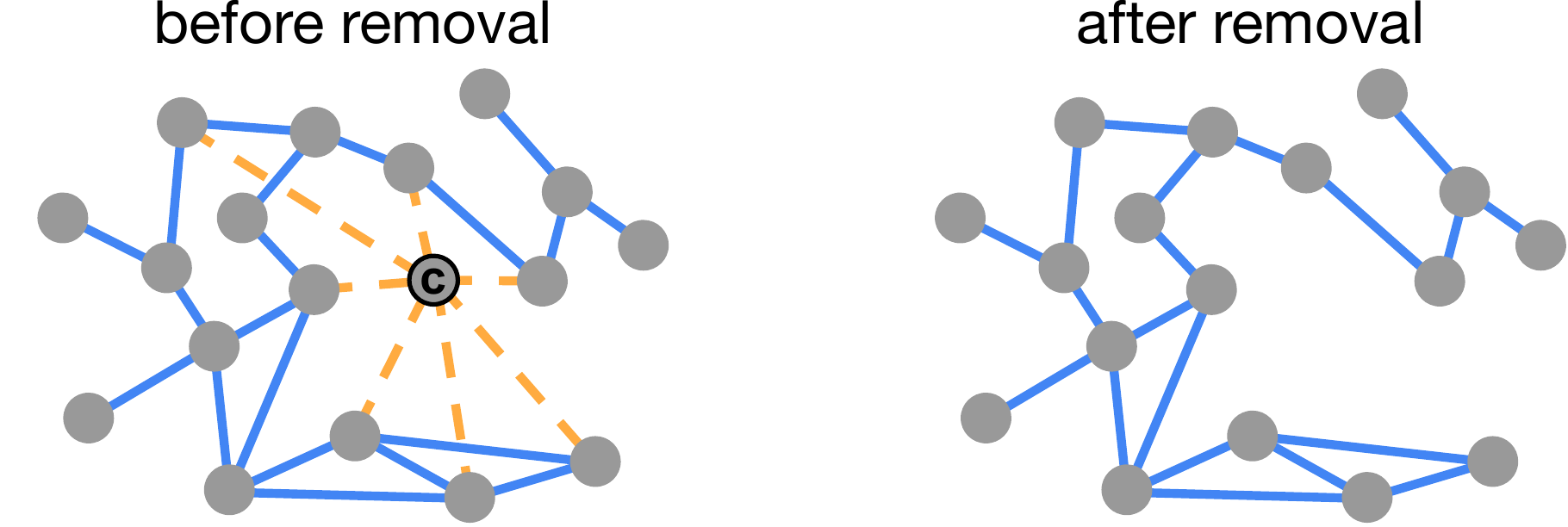}
    
    \begin{flalign*}
    \quad
    B = \left( \begin{array}{c;{2pt/2pt}c} \begin{array}{ccc}  &   & \\  & B' & \\  &   & \\   \end{array}  & D \\  \hdashline[2pt/2pt] E & F \\ \end{array} \right) \hspace{-2.5em} \begin{tabular}{l} $\left.\phantom{\begin{array}{c} \\ D \\ \\ \end{array}}\right\} 2m - 2d $ \\ $\left.\phantom{\begin{array}{c} F \end{array}}\right\} 2d $ \end{tabular}
    \quad\quad\quad\quad
    \left( \begin{array}{ccc}  &   & \\  & B' & \\  &   & \\   \end{array} \right) 
    \quad\quad\quad
    \end{flalign*}
    
    \caption{\emph{Top: } Graph $G$ with target node $c$ before and after removal. $G$ has $m$ edges and $c$ has degree $d$. Dashed yellow edges are incident to $c$, all other edges in solid blue. \emph{Bottom: } Corresponding NB-matrices before and after removal.}
    \label{fig:blue-yellow}
\end{figure}

\subsection{The characteristic polynomial}\label{sec:characteristic-polynomial}
The NB-eigenvalues are the roots of the characteristic polynomial $\det \left( B - t I \right)$. The theory of Schur complements gives us an identity for the determinant of a block matrix,
\begin{align}
\det \left( B- tI \right) &= 
\left|
\begin{array}{cc}
B' - tI & D \\
E & F - tI
\end{array}
\right| \\
&= \det \left( F - tI \right) \det \left( B' - tI - D \left( F- tI \right)^{-1} E \right),
\label{eqn:schur}
\end{align}

where the size of $I$ is given by context. This formula holds whenever $\left( F - tI \right)$ is invertible, i.e., whenever $t$ is not an eigenvalue of $F$. To simplify this expression, we make the following observations.

\begin{lemma}\label{lem:de-ff}
Let $d$ be the degree of target node $c$. With $D, E, F$ as in Equation~(\ref{eqn:block-form}), we have $DE = 0$ and $F^2 = 0$. Therefore, $F$ is nilpotent, that is, all its eigenvalues are zero, and hence $\det \left( F - tI \right) = t^{2d}$. Finally, we have $\left( F - tI \right)^{-1} = - \left( F + tI \right) / t^2$ when $t \neq 0$.
\end{lemma}

\begin{proof}
Since $D, E, F$ are sub-matrices of $B$, their element is given by Equation~(\ref{eqn:nbm-element}). Hence, computing $DE$ and $F^2$ is straightforward, as long as care is placed in keeping track of the appropriate indices for the rows and columns of the involved matrices. Now, $F^2 = 0$ implies that all its eigenvalues are zero and that $\det \left( F - tI \right) = t^{2d}$. Finally, one can manually check that $\left( F - tI \right) \left( F + tI \right) = - t^2 I$.
\end{proof}

Now define $X = DFE$. One can manually check that 
\begin{equation}\label{eqn:x}
X_{k \to l, i \to j} = a_{ck}a_{cj} \left( 1 - \delta_{kj} \right).
\end{equation}
Per the Lemma, Equation~(\ref{eqn:schur}) holds for $t \neq 0$ and
\begin{align}
\det \left( B - tI \right) &= t^{2d} \det \left( B' - tI + \frac{DFE}{t^2} + \frac{DE}{t} \right) \\
&= t^{2d} \det \left( B' - tI + \frac{X}{t^2} \right).
\label{eqn:b-prime-x}
\end{align}

\begin{theorem}\label{thm:char-poly}
For a graph $G$ and target node $c$, suppose the NB-matrix of $G$ is $B$ and the NB-matrix after removing $c$ is $B'$, and let $X$ be as in Equation~(\ref{eqn:b-prime-x}). If $t$ is not an eigenvalue of $B'$ then we have
\begin{equation}\label{eqn:char-poly}
\frac{\det \left( B - tI \right)}{\det \left( B' - tI \right)} = t^{2d} \det \left( I + \frac{1}{t^2} \left( B' - tI \right)^{-1} X\right).
\end{equation}
\end{theorem}

\begin{proof}
Immediate from (\ref{eqn:b-prime-x}) by factoring out $B' - tI$.
\end{proof}

If $c$ has degree $1$, then $X$ equals the zero matrix, and Equation~(\ref{eqn:char-poly}) simplifies to show that removing $c$ has no influence on the non-zero NB-eigenvalues. There are other nodes whose removal do not influence the non-zero NB-eigenvalues, which are characterized as follows. Let the \emph{$2$-core} of $G$ be the graph that remains after iteratively removing nodes of degree $1$. Let the \emph{$1$-shell} of $G$ be the graph induced by the nodes outside of the $2$-core.

\begin{corollary}\label{cor:degree-one}
If $c$ is in the $1$-shell of $G$, removing it does not change the non-zero NB-eigenvalues.
\end{corollary}

\begin{proof}
If $c$ has degree $1$, Equation~(\ref{eqn:x}) gives $X=0$. In this case, (\ref{eqn:char-poly}) becomes $\det \left( B - tI \right) = t^{2d} \det \left( B' - tI \right)$, which implies the assertion. In general, if $c$ is in the $1$-shell, then it must have degree $1$ after iteratively removing some sequence of nodes each of which has degree $1$ at the time of removal. Each of these removals has no effect on the non-zero eigenvalues, and therefore neither does the removal of $c$.
\end{proof}

\begin{remark}
Intuitively, since $F$ is the NB-matrix of a star graph, which contains no NB-walks of length $3$ or more, then immediately we have $F^2 = 0$. Following Figure~\ref{fig:blue-yellow}, $F^2$ counts the number of NB-walks of length $3$ whose edges are yellow-yellow-yellow, of which there are none. Similarly, $DE$ counts the NB-walks whose edges are blue-yellow-blue, which also do not exist. Finally, $X=DFE$ counts the NB-walks of color blue-yellow-yellow-blue, which are precisely those that are destroyed when removing $c$. It is then no surprise that the rest of our analysis pivots fundamentally on the matrix $X$.
\end{remark}

\subsection{The largest eigenvalue}\label{sec:largest-eigenvalue}

We now pivot to study the eigen-drop induced by removing $c$. The larger this eigen-drop, the more influential the target node is in determining the epidemic or percolation thresholds.

\begin{theorem}
With the same assumptions as in Theorem~\ref{thm:char-poly}, let $\lambda_1$ be the largest eigenvalue of $B$ and let $\mathbf{w}$ be a vector such that in Equation~(\ref{eqn:b-prime-x}) we have
\begin{equation}\label{eqn:eigenvector}
\left( B' - \lambda_1 I + \frac{X}{\lambda_1^2} \right) \mathbf{w} = 0.
\end{equation}
Suppose $\{\mathbf{v}_i\}$ is a basis of right eigenvectors of $B'$ and write $\mathbf{w}$ in this basis, $\mathbf{w} = \sum_i w_i \mathbf{v}_i$. Let $\mathbf{u}_1$ be the left eigenvector of $B'$ corresponding to $\mathbf{v}_1$, and set $\alpha_i = \mathbf{u}_1^T X \mathbf{v}_i$. Finally, let $\lambda_1'$ be the largest eigenvalue of $B'$, so that the eigen-drop induced by $c$ is $\lambda_1 - \lambda_1'$ . Then, we have
\begin{align}\label{eqn:main-result}
\lambda_1 - \lambda_1' = \frac{1}{\lambda_1^2} \sum_{i} \frac{w_i}{w_1} \alpha_i.
\end{align}
\end{theorem}

\begin{proof}
If $\mathbf{v}_i$ corresponds to the eigenvalue $\lambda_i'$, then (\ref{eqn:eigenvector}) gives
\begin{align}\label{eqn:bar-baz}
\sum_i w_i \left( B' - \lambda_1 I + \frac{X}{\lambda_1^2} \right) \mathbf{v}_i &= 
\sum_i w_i \left( \lambda_i' \mathbf{v}_i - \lambda_1 \mathbf{v}_i + \frac{X \mathbf{v}_i}{\lambda_1^2} \right) = 0.
\end{align}

Let $\mathbf{u}_1$ be the left eigenvector corresponding to $\mathbf{v}_1$ normalized such that $\mathbf{u}_1^T \mathbf{v}_1 = 1$. Recall that $\mathbf{u}_1$ is orthogonal to every right eigenvector corresponding to a different eigenvalue. Since $\lambda_1'$ has multiplicity one, we have $\mathbf{u}_1^T \mathbf{v}_i = 0$ for each $i \neq 1$. Multiply by $\mathbf{u}_1$ on the left to get

\begin{align}\label{eqn:true-diff}
w_1 \left(\lambda_1' - \lambda_1 \right) + \sum_i w_i \frac{\mathbf{u}_1^T X \mathbf{v}_i}{\lambda_1^2} = 0.
\end{align}

Define $\alpha_i = \mathbf{u}_1^T X \mathbf{v}_i$ and rearrange to get Equation~(\ref{eqn:main-result}).
\end{proof}

\begin{remark}
We can reverse our argument and interpret (\ref{eqn:main-result}) in terms of node addition rather than removal. Suppose the original graph does not contain $c$, and therefore its NB-matrix is $B'$. Then, the NB-matrix \emph{after adding node $c$} is given by (\ref{eqn:block-form}). All our arguments are valid in this setting, and (\ref{eqn:main-result}) then says that the new largest NB-eigenvalue is the solution to a third-degree polynomial, the coefficients of which depend on the full eigendecomposition of $B'$.
\end{remark}

\subsubsection{An approximation}\label{sec:first-approx}
Unfortunately, Equation~(\ref{eqn:main-result}) requires knowledge of all eigenvectors of $B'$. However, in our experience, the vector $\mathbf{w}$ is extremely closely aligned to $\mathbf{v}_1$ and therefore the coefficients $w_i / w_1 \ll 1$. In this case, all but one term in the right-hand side of Equation~(\ref{eqn:main-result}) can be neglected and we get
\begin{equation}\label{eqn:first-approx}
\lambda_1^2 \left( \lambda_1 - \lambda_1' \right) - \alpha_1 \approx 0.
\end{equation}

Here, the larger $\alpha_1$, the larger the eigen-drop $\lambda_1 - \lambda_1'$. Therefore, we study the significance of $\alpha_1$ next.

\begin{proposition}\label{lem:x-nb-centrality}
Let $\mathbf{u}_1, \mathbf{v}_1$ be the left and right eigenvectors of $B'$ normalized such that $\mathbf{u}_1^T \mathbf{v}_1 = 1$. Then we have
\begin{equation}\label{eqn:x-nb-centrality}
\alpha_1 = \mathbf{u}_1^T X \mathbf{v}_1 = \mathbf{v}_1^T P X \mathbf{v}_1 = \left( \sum_i a_{ci} \mathbf{v}_1^i \right)^2 - \sum_i a_{ci} \left( \mathbf{v}_1^i \right)^2,
\end{equation}
where $\mathbf{v}_1^i$ is the NB-centrality of node $i$ in the graph after removal (see Equation~(\ref{eqn:nb-centrality})). We call $\alpha_1$ the \emph{\mbox{$X$-non-backtracking centrality}}, or \emph{\Xnb{} centrality}, of $c$.
\end{proposition}

\begin{proof}
The first equality comes from the fact that $\mathbf{u}_1 = P \mathbf{v}_1$, by Lemma~\ref{lem:right-left}. We can find $PX_{k \to l, i \to j} = a_{cl}a_{cj} \left( 1 - \delta_{lj} \right)$ using Equations (\ref{eqn:x}) and (\ref{eqn:p}). The result then follows from manually computing $\mathbf{v}_1^T PX \mathbf{v}_1$ and applying Equation~(\ref{eqn:nb-centrality}).
\end{proof}

The Proposition establishes that the behavior of the eigen-drop in (\ref{eqn:first-approx}) is governed by the \Xnb{} centrality of $c$ in (\ref{eqn:x-nb-centrality}), which is a function only of the NB-centralities of $c$'s neighbors. Importantly, these centralities are measured \emph{after} $c$ is removed. We come back to this point in Section~\ref{sec:immunization}. Notably, the principal eigenvector is normalized by $\mathbf{u}_1^T \mathbf{v}_1 = \mathbf{v}_1^T P \mathbf{v}_1 = 1$, i.e. it does not have unit length.

\subsubsection{An upper bound}\label{sec:upper-bound}
An alternative way of studying the eigen-drop is by choosing $\mathbf{w}$ such that $w_1 = 1$, and bounding
\begin{equation}
q = q(c) = \sum_i w_i \alpha_i,
\end{equation}
which drives the right-hand side of Equation~(\ref{eqn:main-result}).

Suppose that $R$ is the matrix whose columns are the eigenvectors $\{\mathbf{v}_i\}$, and let $L = R^{-1}$ such that $B' = R \Lambda L$, where $\Lambda$ is the diagonal matrix of the eigenvalues $\{\lambda_i'\}$. The rows of $L$ are left eigenvectors of $B'$, in particular, $\mathbf{u}_1^T$ is the first row of $L$. Then we have $\alpha_i = \left( L X R \right)_{1i}$, and $q$ is the dot product between the first row of $LXR$ and $\mathbf{w}$,
\begin{equation}
q = \mathbf{e_1^T} L X R \mathbf{w} = \Tr \left( L X R \, \mathbf{w} \mathbf{e_1^T} \right),
\end{equation}
where $\mathbf{e_1} = \left( 1, 0, \ldots, 0 \right)$. Using the cyclic property of the trace, and the fact that $P^2 = I$, we now have
\begin{align}
q &= \Tr \left( L X R \, \mathbf{w} \mathbf{e_1^T} \right) = \Tr \left( X R \, \mathbf{w} \mathbf{e_1^T} \, L \right) = \Tr \left( P X R \, \mathbf{w} \mathbf{e_1^T} \, L P \right).
\end{align}
Applying the Cauchy-Schwarz inequality for the trace gives us
\begin{align}
q &\leq \left| PX \right|_F \, \left| R \, \mathbf{w} \mathbf{e_1^T} \, L P \right|_F,
\end{align}
where $\left| M \right|_F^2 = \Tr\left( M^T M \right)$ is the Frobenius norm. Finally, the fact that $\mathbf{w} \mathbf{e_1^T}$ is a matrix with rank one gives 
\begin{align}
q &\leq \left| PX \right|_F \, \left( \mathbf{e_1^T} L P R \mathbf{w} \right).
\end{align}
As before, we have $w_i / w_1 \ll 1$ and since we chose $w_1 = 1$, the term $\left( \mathbf{e_1^T} L P R \mathbf{w} \right)$ is very close to $1$. Therefore, we obtain $|PX|_F$ as an (approximate) upper bound for $q$. Observe that since $PX$ is non-negative, we have $|PX|_F = \one^T PX \one$, where $\one = \left( 1, 1, \ldots, 1\right)$.

\begin{proposition}\label{lem:x-degree-centrality}
In Equation~(\ref{eqn:main-result}), let $\mathbf{w}$ be such that $w_1 = 1$, and define $q = \sum_i w_i \alpha_i$. The quantity $\one^T PX \one$ is an approximate upper bound for $q$, that is, $q \leq \one^T PX \one \left( \mathbf{e_1^T} L P R \mathbf{w} \right)$. Furthermore, we have
\begin{equation}\label{eqn:x-degree-centrality-minus-one}
\one^T PX \one = \left( \sum_i a_{ci} \left(d_i - 1 \right) \right)^2 - \sum_i a_{ci} \left(d_i - 1\right)^2,
\end{equation}
where $d_i$ is the degree of node $i$ in the original graph, before removal.  We call $\one^T PX \one$ the \emph{\Xdeg{} centrality} of $c$.
\end{proposition}
\begin{proof}
The first claim was proved in the previous paragraphs. The second claim comes from direct evaluation of $\one^T PX \one$ using $PX_{k \to l, i \to j} = a_{cl}a_{cj} \left( 1 - \delta_{lj} \right)$, keeping in mind the degrees are measured after removal.
\end{proof}

\begin{remark}
The fact that the \Xdeg{} of $c$, $\one^T PX \one$, bounds $q$ only approximately merits further theoretical consideration. However, it will be immaterial in our exposition going forward, as our experiments will show that, in practice, the \Xdeg{} of nodes is an excellent predictor of the node's eigen-gap, regardless of the value of $\mathbf{e_1^T} L P R \mathbf{w}$.
\end{remark}

\subsection{X-centrality}\label{sec:x-centrality}
In Section~\ref{sec:first-approx} we use the \Xnb{} centrality, $\mathbf{v}_1^T PX \mathbf{v}_1$, while in Section~\ref{sec:upper-bound} we use the \Xdeg{} centrality, $\mathbf{1}^T PX \mathbf{1}$, both for the purpose of studying the eigen-drop induced by $c$. The former is a function of the NB-centralities of the neighbors of $c$ (Proposition~\ref{lem:x-nb-centrality}), while the latter is a function of their degrees (Proposition~\ref{lem:x-degree-centrality}). Importantly, both centralities are measured \emph{after} $c$ has been removed.

Consider a fixed target node $c$, which in turn fixes $X$ and $P$. The matrix $PX$ is capable of defining new node-level statistics given a vector of values for each directed edge. It does so by aggregating the edge values along NB-walks that go through $c$; following Figure~\ref{fig:blue-yellow}, this aggregation is done along blue-yellow-yellow-blue walks. Recall that if $G$ has $m$ (undirected) edges and $c$ has degree $d$, then $X$ and $P$ are of size $2m - 2d$. Given an arbitrary vector $\mathbf{z}$ of size $2m - 2d$, we have
\begin{equation}\label{eqn:quadratic-form}
\mathbf{z}^T PX \mathbf{z} = \left( \sum_i a_{ci} \sum_j \mathbf{z}_{j \to i} \right)^2 - \sum_i a_{ci} \left( \sum_j \mathbf{z}_{j \to i}\right)^2.
\end{equation}

One can evaluate the right-hand side of (\ref{eqn:quadratic-form}) for any vector $\mathbf{z}$ of size $2m$, and use only the $2m - 2d$ entries that correspond to edges not incident to $c$. In other words, we do not need to know $X$ or $P$, but only who the neighbors of $c$ are. Since $c$ determines both $X$ and $P$, the same vector $\mathbf{z}$ can be evaluated using different target nodes. Therefore the quantity in (\ref{eqn:quadratic-form}) naturally corresponds to whichever target node was used to evaluate it, and can be thought of as a node-level quantity derived from $\mathbf{z}$.

Now define $\mathbf{z}^i = \sum_j \mathbf{z}_{j \to i}$ and let $\Var_c \left( \mathbf{z}^i \right)$ be the variance of the $\mathbf{z}^i$ values corresponding to neighbors of $c$. Then we have
\begin{equation}
\Var_c \left( \mathbf{z}^i \right) = \frac{\sum_i a_{ci} \left( \mathbf{z}^i \right)^2}{d} - \left( \frac{\sum_i a_{ci} \mathbf{z}^i}{d} \right)^2,
\end{equation}
which differs from (\ref{eqn:quadratic-form}) only in sign and a (non-linear) normalization. Accordingly, $\mathbf{z}^T PX \mathbf{z}$ will have large values when $\mathbf{z}^i$ has little variability among the neighbors of $c$.

Using this framework we could define, for example, \emph{\mbox{$X$-closeness} centrality}, \emph{\mbox{$X$-betweenness} centrality}, etc. Whether these concepts are as useful as the two studied here remains an open question.

\section{Node immunization}\label{sec:immunization}
Targeted immunization works as follows. Given a graph $G$ and an integer $p$, we want to remove from $G$ the $p$ nodes that increase the epidemic threshold the most (equivalently, decrease the largest NB-eigenvalue the most). Common strategies involve three steps: (i) the nodes are sorted by decreasing values of a certain statistic, for example degree; (ii) the node with the highest value of this statistic is removed from the graph; and (iii) the statistic has to be recomputed after each time a node is removed. These steps are repeated until the target number $p$ has been removed. In this context, our framework presents two major obstacles:
\begin{enumerate}[label=\alph*)]
\item Both the \Xnb{} and \Xdeg{} centralities of a node must be computed \emph{after} the node has been removed. So, to execute the step (i) above, we need to temporarily remove each node in turn before we decide which one to ultimately remove, which defeats the purpose of targeted immunization.
\item For step (iii), we must guarantee that recomputing the statistic of every node at each step is an efficient procedure.
\end{enumerate}

\subsection{Using X-NB centrality}\label{sec:using-x-nb}
Algorithm~\ref{alg:naive-xnb} naively follows the steps above to implement an immunization strategy based on \Xnb{} centrality. We are tempted to think this strategy is the ``right'' one, as it approximates the true effect of a node's removal in the epidemic threshold. However, we must address the obstacles mentioned above.

\begin{algorithm}
\caption{Naive \Xnb{} immunization strategy.}
\label{alg:naive-xnb}
\SetKwFunction{RemoveNode}{RemoveNode}
\SetKwFunction{NBCentrality}{NBCentrality}
\SetKwFunction{XNBCentrality}{XNBCentrality}
\SetKwFunction{length}{length}
\SetKwFunction{AuxNBMatrix}{AuxNBMatrix}
\SetKwData{removed}{removed}
\SetKwData{node}{node}
\SetKwData{XNB}{XNB}
\SetKwData{XNBi}{XNB[i]}
\SetKwData{XNBc}{XNB[c]}
\SetKwData{i}{i}
\SetKwData{c}{c}
\SetAlgoNoEnd
\DontPrintSemicolon
\KwIn{graph $G$, integer $p$}
\KwOut{\removed, an ordered list of nodes to immunize}
\BlankLine
\removed $\leftarrow \emptyset$ \;
\XNBi $\leftarrow 0$ for each node \i \;
\While{\length{\removed} $< p$}{
  \ForEach{node \c in $G$}{
    $H \leftarrow$ \RemoveNode{$G$, \c} \;
    $v_H \leftarrow$ principal eigenvector of \AuxNBMatrix{$H$} \;
    \XNBc $\leftarrow$ \XNBCentrality{$v_H$, \c} \;
  }
  \node $\leftarrow \argmax_i$ \XNBi \;
  $G \leftarrow$ \RemoveNode{$G$, \node} \;
  \removed.append(\node) \;
}
\Return{\removed}
\end{algorithm}

\begin{algorithm}
\caption{Approximate \Xnb{} immunization strategy.}
\label{alg:approx-xnb}
\SetKwFunction{RemoveNode}{RemoveNode}
\SetKwFunction{NBCentrality}{NBCentrality}
\SetKwFunction{XNBCentrality}{XNBCentrality}
\SetKwFunction{length}{length}
\SetKwData{removed}{removed}
\SetKwData{node}{node}
\SetKwData{XNB}{XNB}
\SetKwFunction{AuxNBMatrix}{AuxNBMatrix}
\SetKwData{XNBi}{XNB[i]}
\SetKwData{XNBc}{XNB[c]}
\SetKwData{i}{i}
\SetKwData{c}{c}
\SetAlgoNoEnd
\DontPrintSemicolon
\KwIn{graph $G$, integer $p$}
\KwOut{\removed, an ordered list of nodes to immunize}
\BlankLine
\removed $\leftarrow \emptyset$ \;
\XNBi $\leftarrow 0$ for each node \i \;
\While{\length{\removed} $< p$}{
  $v_G \leftarrow$ principal eigenvector of \AuxNBMatrix{$G$} \;
  \ForEach{node \c in $G$}{
    \XNBc $\leftarrow$ \XNBCentrality{$v_G$, \c} \;
  }
  \node $\leftarrow \argmax_i$ \XNBi \;
  $G \leftarrow$ \RemoveNode{$G$, \node} \;
  \removed.append(\node) \;
}
\Return{\removed}
\end{algorithm}

To overcome obstacle (a), we propose to approximate Equation~(\ref{eqn:x-nb-centrality}) by using the NB-centralities in the original graph before removing any node even temporarily. Algorithm~\ref{alg:approx-xnb} takes this approximation into account. The error incurred by this approximation is dampened by the fact that what we are ultimately interested in is the ranking of the nodes rather than the actual values of their centralities. For obstacle (b), one could use a strategy similar to \cite{ChenTPTEFC16}, where they devise an algorithm to approximate the impact on a node's eigenvector centrality after the removal of a node without having to recompute the values again. However, doing so for $X$-NB centrality remains an open question.

\subsubsection*{Complexity Analysis}\label{sec:complexity-1}
We assume that $G$ is given in adjacency list format. In Algorithm~\ref{alg:naive-xnb}, lines $2$ and $8$ take $n$ operations each. Line $5$ creates a copy $H$ of the adjacency list and removes the target node $c$ from it (but leaves $G$ intact). Line $6$ uses Equation~(\ref{eqn:aux-nb-matrix}) to compute the auxiliary NB-matrix, which takes $O(m)$ time, and it takes $O\left( m \right)$ to compute the principal eigenvector (using, e.g. the Lanczos algorithm with a number of iterations that does not depend on the parameters). Line $7$ uses Lemma~\ref{lem:small-big} to compute the correctly normalized NB-centralities, and Equation~(\ref{eqn:x-nb-centrality}) to compute \Xnb{} centralities, both of which take $n$ operations. The remaining lines take constant time. Accounting for loops, Algorithm~\ref{alg:naive-xnb} takes a total of $O\left(n + p \left( n \left( m + n \right) + n \right) \right) = O\left( pn \left( m + n \right) \right)$. In Algorithm~\ref{alg:approx-xnb}, the NB-centralities are computed outside of the inner loop, which gives a complexity of $O\left( p \left( m + n \right) \right)$, or $O\left( m + n \right)$ for constant $p$.

\subsection{Using X-degree}\label{sec:using-x-deg}
\Xdeg{} can be easily computed without temporarily removing any nodes, see Equation~(\ref{eqn:x-degree-centrality-minus-one}). Indeed, all we need to know about the graph after removal is the degree of each node. Hence, obstacle (a) is easily overcome in this case. Further, after each step we need not recompute the \Xdeg{} of all nodes, but only of those nodes two steps away from the target node. Indeed, removing $c$ changes the degree of its neighbors, which in turn changes the \Xdeg{} of its neighbors' neighbors. So obstacle (b) is also overcome. Algorithm~\ref{alg:xdeg} implements this strategy. Importantly, it does not involve the computation of any matrices or their eigenvectors.

\subsubsection*{Complexity Analysis}\label{sec:complexity-2}
Lines $1,6,10,11$ of Algorithm~\ref{alg:xdeg} take constant time, while line $2$ takes $O(m)$. When using a standard map (or dictionary) to store the \Xdeg{} values, line $4$ takes $O(n)$ operations, and line $9$ takes $O(1)$. Now suppose that the nodes removed by Algorithm~\ref{alg:xdeg} are, in order, $i_1,\ldots,i_p$. At iteration $j$, the loop in line $5$ takes $d_{i_j}$ operations, and the double loop in lines $7-8$ takes as many iterations as the number of nodes two steps away from $i_j$, say $D_{i_j}$. This yields a total of $O \left( m + p n + \sum_{j=1}^p d_{i_j} + \sum_{j=1}^p D_{i_j} \right)$. We can also implement Algorithm~\ref{alg:xdeg} using an indexed priority queue (IPQ) to store the \Xdeg{} values instead of a map; see Appendix~\ref{app:complexity}. In this case the worst case scenario complexity is $O \left( m + p \log n + \sum_{j=1}^p d_{i_j} + \log n \sum_{j=1}^p D_{i_j} \right)$. In Appendix~\ref{app:complexity} we refine this analysis for networks with homogeneous or heterogeneous degree distributions, and show that the map or IPQ versions have better worst case scenario scalability for different values of network parameters.

Importantly, the average runtime of both versions is in fact close to linear, with the IPQ version being the fastest. Figure~\ref{fig:scaling} shows the average runtime of both versions on random power-law configuration model graphs with varying degree exponent $\gamma$ and constant $p$ (see Appendix~\ref{app:complexity} for details). The reason the average runtime is considerably faster than the worst-case scenario is because graphs typically have very few large hubs. That is, roughly speaking, there are $O(1)$ many nodes that take $O(n)$ time to process, while there are $O(n)$ many nodes that take $O(1)$ time to process. This effect is intensified the closer $\gamma$ is to $2$, which counterbalances the exponent $\frac{2}{\gamma - 1}$ in the worst case scenario.

\begin{algorithm}
\caption{\Xdeg{} immunization strategy.}
\label{alg:xdeg}
\SetKwFunction{RemoveNode}{RemoveNode}
\SetKwFunction{XDegree}{XDegree}
\SetKwFunction{length}{length}
\SetKwData{removed}{removed}
\SetKwData{node}{node}
\SetKwData{neighbors}{neighbors}
\SetKwData{neighborsnode}{neighbors[node]}
\SetKwData{neighborsi}{neighbors[i]}
\SetKwFunction{AuxNBMatrix}{AuxNBMatrix}
\SetKwFunction{heapify}{heapify}
\SetKwData{Xdeg}{XDeg}
\SetKwData{Xdegi}{XDeg[i]}
\SetKwData{Xdegj}{XDeg[j]}
\SetKwData{Xdegc}{XDeg[c]}
\SetKwData{queue}{queue}
\SetKwData{i}{i}
\SetKwData{j}{j}
\SetKwData{c}{c}
\SetAlgoNoEnd
\DontPrintSemicolon
\KwIn{graph $G$, integer $p$}
\KwOut{\removed, an ordered list of nodes to immunize}
\BlankLine
\removed $\leftarrow \emptyset$ \;
\Xdegi $\leftarrow$ \XDegree{$G$, \i} for each node \i \;
\While{\length{\removed} $< p$}{
  \node $\leftarrow \max_i $ \Xdegi \;
  \ForEach{\i in $G$.\neighborsnode}{
    $G$.\neighborsi.remove(\node) \;
  }
  \ForEach{\i in $G$.\neighborsnode}{
    \ForEach{\j in $G$.\neighborsi}{
      \Xdegj $\leftarrow$ \XDegree{$G$, \j}
    }
  }
  $G$.\neighborsnode $\leftarrow \emptyset$ \;
  \removed.append(\node) \;
}
\Return{\removed}
\end{algorithm}

\begin{figure}
    \centering
    \includegraphics[width=0.75\columnwidth]{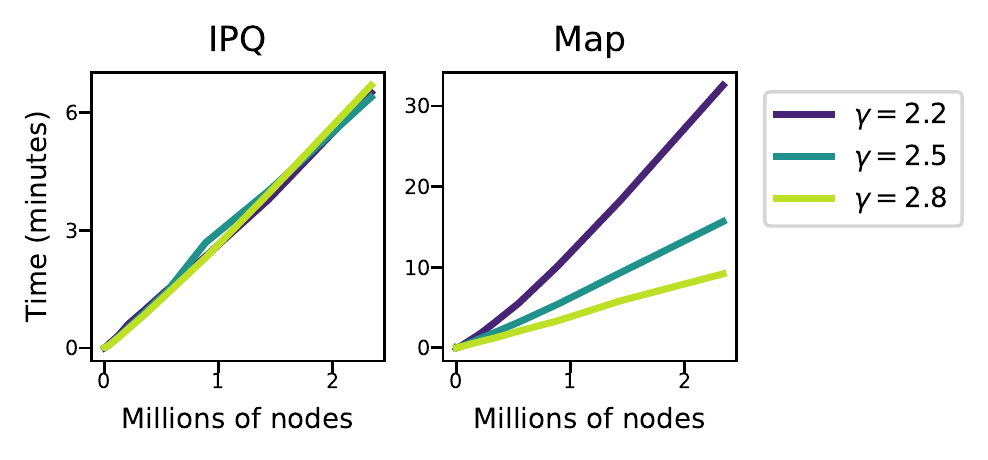}    
    \caption{Average runtime scaling of Algorithm~\ref{alg:xdeg} on random power-law graphs with varying degree exponent $\gamma$. The runtimes are linear, with IPQ being faster than Map.}  
    \label{fig:scaling}
\end{figure}


\section{Related work}\label{sec:related-work}
\paragraph{Perturbation of NB-matrix} \citet{zhang2014non} briefly treats the case of eigenvalue perturbation of a matrix derived from the NB-matrix in the case of edge removal, while \citet{abs-1907-05603} analyze the perturbation of quadratic eigenvalue problems, with applications to the NB-eigenvalues of the stochastic block model. Our theory is more general since it studies node removal (as opposed to single edge removal), and it applies to any arbitrary graph.

\paragraph{NB centrality} Many notions of centrality based on the NB-matrix exist, for example NB-PageRank \citep{ArrigoHN19}, NB-centrality \citep{martin2014localization,radicchi2016leveraging}, and Collective Influence \citep{morone2015influence,morone2016collective}. The latter two have been proposed as solutions to the problem of ``influencer identification''. This problem aims to find nodes that determine the course of spreading dynamics, and is thus more general than our objective of increasing the epidemic threshold. Collective Influence in particular is similar to \Xdeg{}; see Appendix~\ref{app:ci}. Also in this context, \citet{kitsak2010identification} propose to use the $k$-core index, and \citet{abs-1912-08459} highlight the importance of node degree. We compare our algorithms to all of these baselines in Section~\ref{sec:experiments}. Finally, \citet{EverettB10} study the influence of a node's removal in other nodes' centrality, which is reminiscent to our $X$-centrality framework.

\paragraph{Targeted immunization} \citet{pastor2015epidemic} review general immunization strategies and other generalities of spreading dynamics on networks. \citet{ChenTPTEFC16} propose \texttt{NetShield}, an efficient algorithm for immunization focusing on decreasing the largest eigenvalue of the adjacency matrix. We prefer to focus on decreasing the largest NB-eigenvalue instead since it provides a tighter bound to the true epidemic threshold in certain cases \citep{karrer2014percolation,hamilton2014tight,shrestha2015message}. \citet{LinCZ17} study the percolation threshold in terms of so-called high-order non-backtracking matrices. Percolation thresholds are tightly related to epidemic thresholds of SIR dynamics \citep{pastor2015epidemic,newman2002spread}.


\section{Experiments}\label{sec:experiments}

\subsection{Approximating the Eigenvalue}\label{sec:exp-theoretical-accuracy}
\textbf{How close is the approximation in Equation~(\ref{eqn:first-approx})?} We first compute the largest NB-eigenvalue $\lambda_1$ of a graph $G$. Then we fix a target node $c$ and remove it from $G$ and compute the new eigenvalue $\lambda_c$. (For ease of notation, in this section we use $\lambda_c$ instead of $\lambda_1'$, and $\alpha$ instead of $\alpha_1$.) Finally, we use (\ref{eqn:first-approx}) to compute two approximations,
\begin{equation}
\widehat{\lambda}_c = \lambda_1 - \alpha / \lambda_1^2, \quad\quad
\widetilde{\lambda}_c = \lambda_1 - \widetilde{\alpha} / \lambda_1^2,
\end{equation}
where $\alpha$ is the true \Xnb{} centrality of $c$, and $\widetilde{\alpha}$ is the approximate \Xnb{} centrality used in Algorithm~\ref{alg:approx-xnb}, i.e., it is computed using the NB-centralities before removing $c$. We now compare the approximations $\widehat{\lambda}_c$ and $\widetilde{\lambda}_c$ to the true value of $\lambda_c$ for randomly selected nodes of synthetic graphs. We use different synthetic random graph models: Watts-Strogatz (WS) \citep{watts1998collective}, Stochastic Block Model (SBM) \citep{girvan2002community,karrer2011stochastic}, Barab\'asi-Albert (BA) \citep{albert2002statistical}, and 
Block Two--Level Erd\H{o}s-R\'{e}nyi (BTER) \citep{seshadhri2012community}. See Section~\ref{app:data} for more details on the data sets, and  Section~\ref{app:approximating-the-eigenvalue} for details on the experimental setup.

\begin{figure*}
    \centering
\makebox[\textwidth][c]{\includegraphics[width=1.2\textwidth]{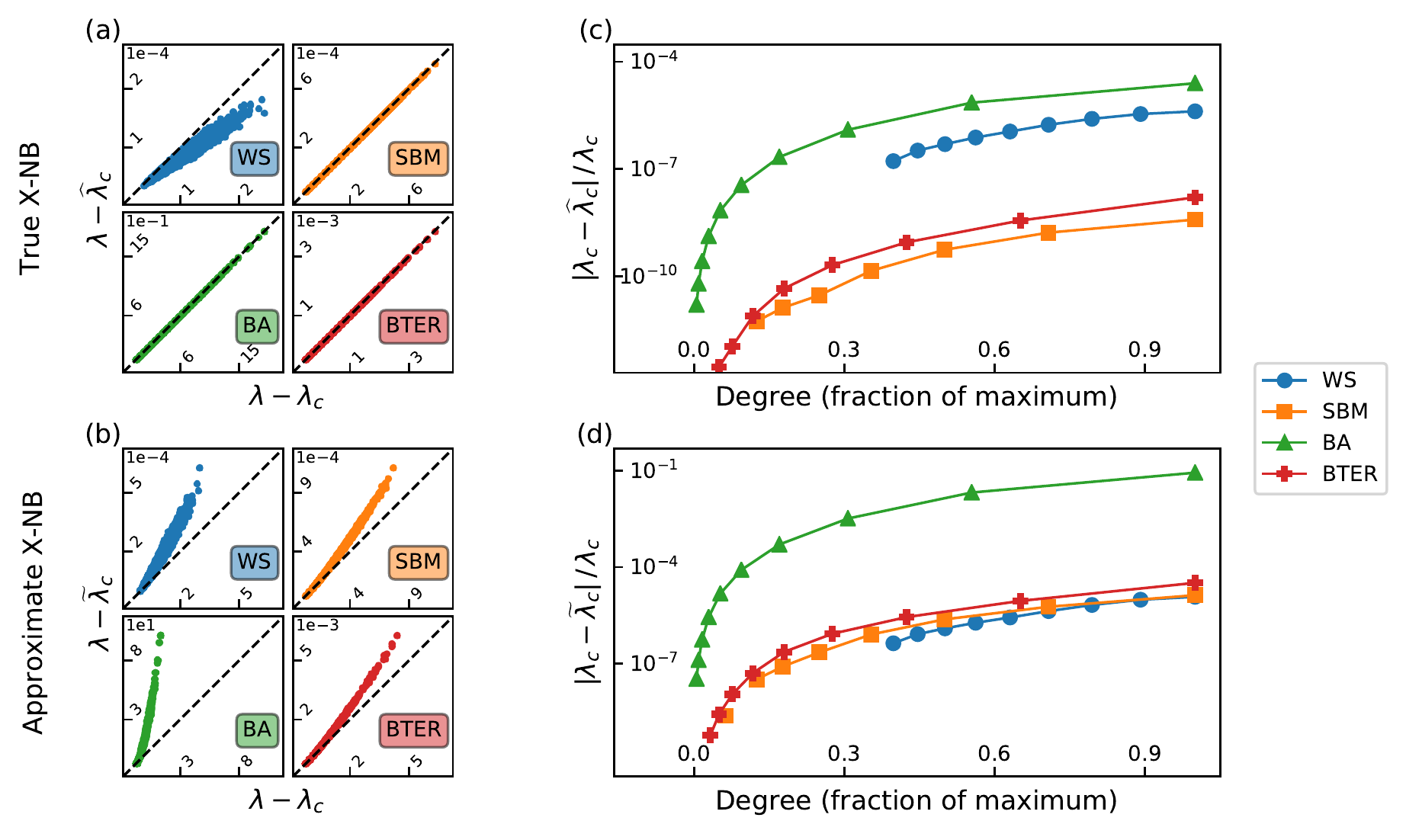}}
    \caption{\emph{Left:} True eigen-drop (vertical axis) vs approx. eigen-drop (horizontal axis), using \emph{(a)} true, and \emph{(b)} approx. values of \Xnb{}. Dashed line is $y = x$. Each marker represents one node. \emph{Right:} Relative error when predicting $\lambda_c$, as a function of degree, using \emph{(c)} true, and \emph{(d)} approx. values of \Xnb{}. Degrees expressed as a fraction of the maximum degree among graphs in the same ensemble. Each marker is the average within log-binned values of degree; error bars too small to show at this scale. WS graphs (blue circles) have no nodes whose degree is less than 30\% of the maximum. Our approximation of the eigen-gap is accurate.}
    \label{fig:th-acc}
\end{figure*}

Fig.~\ref{fig:th-acc}a shows that our approximation is extremely close for all graphs tested, though it tends to underestimate the eigen-drop in WS graphs. Fig.~\ref{fig:th-acc}c shows the average relative error versus degree. Our approximation worsens as degree increases, though it is quite small for most degrees. In the worst case, the relative error is less than $10^{-4}$, or $0.01\%$. Fig.~\ref{fig:th-acc}b shows the eigen-drop computed using the approximate version of \Xnb{}. This approximation is systematically overestimating the true eigen-drop. Fig.~\ref{fig:th-acc}d shows that this systematic error is of the order of $10\%$ in the worst case, though it is negligible for small degrees. 
In all, Figure~\ref{fig:th-acc} confirms the accuracy of our approximations, and it points to the fact that the terms neglected in (\ref{eqn:first-approx}) will become larger as degree increases.

\subsection{Predicting the Eigen-drop}\label{sec:exp-predicting-eigengap}
\textbf{How well can \Xnb{} centrality and \Xdeg{} predict a node's eigen-drop?} Unlike in Experiment \ref{sec:exp-theoretical-accuracy}, here we do not approximate the eigen-drop, but only seek to predict its size. (In fact, we cannot use \Xdeg{} to approximate the eigen-drop at all.) See Section~\ref{app:predicting-the-eigengap} for experimental setup, and Section~\ref{app:data} for details on data sets.

\begin{figure*}
    \centering
    \makebox[\textwidth][c]{\includegraphics[width=1.25\textwidth]{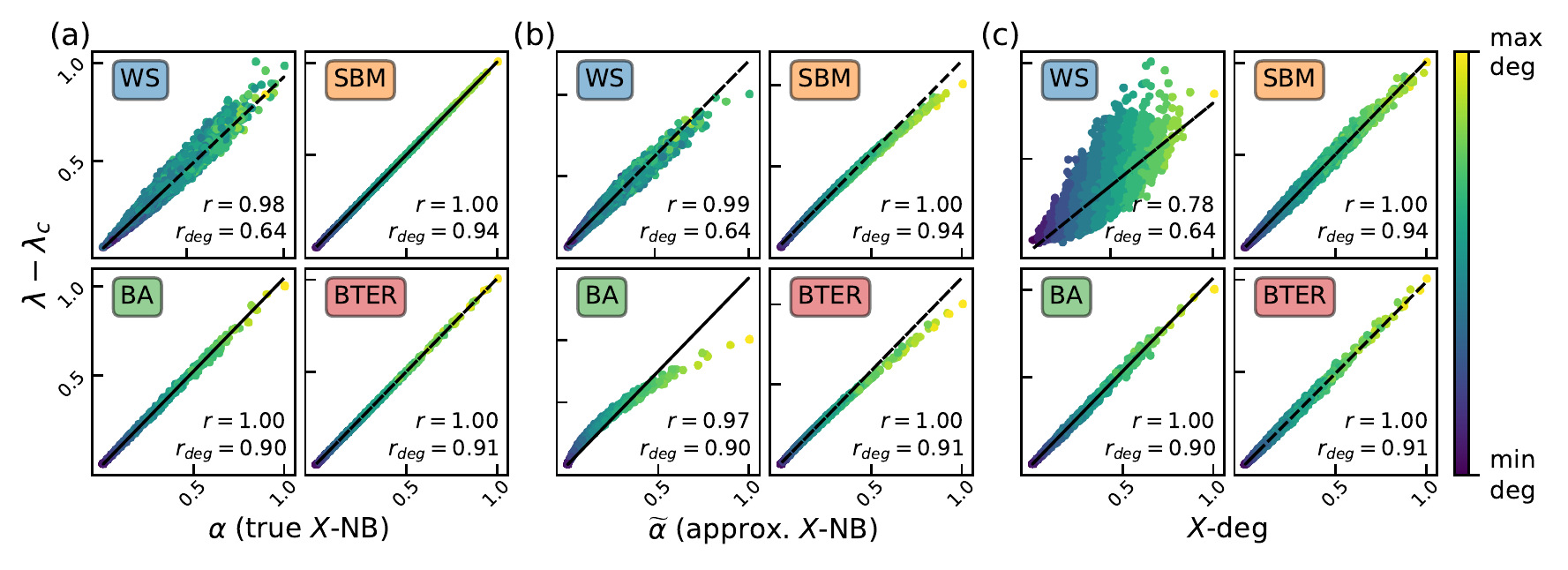}}
    \caption{Predicting the eigen-drop using \emph{(a)} true value of \Xnb{}, \emph{(b)} approx. value of \Xnb{}, and \emph{(c)} \Xdeg{}. Markers colored by degree, expressed as a fraction of the largest degree in the same ensemble. Each panel shows the correlation coefficient between the corresponding statistic and the eigen-drop ($r$), and the correlation between degree and the eigen-drop ($r_{deg}$). Dashed lines are linear regression lines. Our proposed node-level statistics accurately track eigen-drops in random graph models such as BA, SBM, and BTER. \Xdeg{}~underestimates the eigen-drop in WS graphs. }
    \label{fig:predicting}
\end{figure*}

In Fig.~\ref{fig:predicting}a we measure how correlated the true value of \Xnb{}, denoted by $\alpha$, is to the true eigen-drop. For SBM, BA, and BTER graphs, the magnitude of $\alpha$ lines up extremely closely with the value of the eigen-drop, showing a correlation coefficient of $r=1.00$. In all cases, $\alpha$ is better correlated to the eigen-drop than degree (as shown by the correlation coefficients $r_{\deg}$). In WS we see considerably more variance than in other ensembles though $\alpha$ is still an excellent predictor of the eigen-drop, at $r=0.98$. This picture repeats itself when using the approximate value of \Xnb{}, $\widetilde{\alpha}$ (Fig.~\ref{fig:predicting}b), and \Xdeg{} (Fig.~\ref{fig:predicting}c). $\widetilde{\alpha}$ seems to slightly underestimate the eigen-drop, while \Xdeg{} has noticeably more variance than the other two statistics, especially in WS. All three statistics are better correlated to the eigen-drop than degree in all graph ensembles. We highlight that even when some of the panels in Fig.~\ref{fig:predicting} are not precisely linear, they all show that the eigen-drop is an increasing function of all of $\alpha$, $\widetilde{\alpha}$, and \Xdeg{}. These results encourage us to use \Xnb{} and \Xdeg{} as immunization strategies. Further, using $\alpha$ has very little advantage over $\widetilde{\alpha}$, and therefore we are justified in using Algorithm~\ref{alg:approx-xnb} instead Algorithm~\ref{alg:naive-xnb} for computational reasons.

\subsection{Immunization with \mbox{X-NB} and \mbox{X-degree}}\label{sec:yyyyyy}
\textbf{How effective are \Xnb{} and \Xdeg{} at immunization?} We remove $1,2$ and $3$ percent of nodes using different strategies and evaluate the resulting eigenvalue. We use the immunization strategies node degree (\texttt{degree}), $k$-core index (\texttt{core}), NetShield (\texttt{NS}), Collective Influence (\texttt{CI}), NB-centrality (\texttt{NB}), approximate \Xnb{} (\texttt{XNB}), and \Xdeg{}  (\texttt{Xdeg}). For computational reasons, we do not use the true value of \Xnb; for more details on baselines see Section~\ref{app:baselines}. In all data sets, \texttt{core} had the least performance and is therefore not shown in our results. We hypothesize this is because many nodes can have the same $k$-core index at the same time, so \texttt{core} cannot identify which is the best one among all of them.

\begin{table}
\centering
\resizebox{0.75\textwidth}{!}{%
\begin{tabular}{r|cc|cc|cc}
\toprule
{} &  \texttt{degree} &      \texttt{NS} &      \texttt{CI} &    \texttt{Xdeg} &      \texttt{NB} &     \texttt{XNB} \\
\midrule
   1\%   &  62.76 &  61.44 &  62.88 &  62.90 &  62.92 &  62.91 \\
BA 2\%   &  68.84 &  66.94 &  68.97 &  68.99 &  69.01 &  69.01 \\
   3\%   &  72.42 &  70.09 &  72.56 &  72.57 &  72.59 &  72.59 \\
\midrule
     1\% &  ~6.28 &  ~6.40 &  ~6.41 &  ~6.45 &  ~6.46 &  ~6.46 \\
BTER 2\% &  10.60 &  10.72 &  10.80 &  10.85 &  10.86 &  10.86 \\
     3\% &  14.31 &  14.40 &  14.55 &  14.61 &  14.63 &  14.63 \\
\midrule
    1\%  &  ~3.31 &  ~3.41 &  ~3.40 &  ~3.43 &  ~3.44 &  ~3.44 \\
SBM 2\%  &  ~6.00 &  ~6.16 &  ~6.19 &  ~6.23 &  ~6.25 &  ~6.25 \\
    3\%  &  ~8.52 &  ~8.66 &  ~8.76 &  ~8.80 &  ~8.82 &  ~8.82 \\
\midrule
   1\%   &  ~1.41 &  ~1.17 &  ~1.50 &  ~1.52 &  ~1.63 &  ~1.63 \\
WS 2\%   &  ~2.52 &  ~2.09 &  ~2.97 &  ~2.98 &  ~3.11 &  ~3.11 \\
   3\%   &  ~3.66 &  ~2.94 &  ~4.41 &  ~4.41 &  ~4.57 &  ~4.58 \\
\bottomrule
\end{tabular}}

\caption{Average percentage eigen-drop (larger is better) on synthetic graphs after removing 1\%, 2\%, and 3\% of the nodes using different strategies. Strategies are (column) grouped in performance tiers. \texttt{NB} and \texttt{XNB} have the best performances.}
\label{tab:imm}
\end{table}

\begin{table*}
\centering
\resizebox{\textwidth}{!}{%
\begin{tabular}{r|ccc|ccc|ccc}
\toprule
{} & \multicolumn{3}{c}{$p=1$} & \multicolumn{3}{c}{$p=10$} & \multicolumn{3}{c}{$p=100$} \\
{} & \texttt{degree} &     \texttt{CI} &   \texttt{Xdeg} & \texttt{degree} &     \texttt{CI} &   \texttt{Xdeg} & \texttt{degree} &     \texttt{CI} &   \texttt{Xdeg} \\
\midrule
\texttt{AS-1}                 &  ~0.74 &  ~0.74 &  \textbf{~2.35} &  ~6.70 &  13.51 &  \textbf{15.43} &  71.65 &  \textbf{78.26} &  75.92 \\
\texttt{AS-2}                 &  ~2.02 &  ~2.02 &  \textbf{~4.00} &  17.09 &  22.36 &  \textbf{28.17} &  87.60 &  \textbf{89.61} &  87.02 \\
\texttt{Social-Slashdot}      &  ~0.95 &  \textbf{~1.02} &  \textbf{~1.02} &  ~4.63 &  ~6.06 &  \textbf{~6.94} &  23.65 &  28.11 &  \textbf{30.30} \\
\texttt{Social-Twitter}       &  \textbf{~2.18} &  \textbf{~2.18} &  ~1.98 &  13.21 &  \textbf{13.97} &  13.68 &  41.10 &  42.88 &  \textbf{43.39} \\
\texttt{Transport-California} &  ~0.00 &  ~0.00 &  \textbf{~0.65} &  \textbf{~2.65} &  ~0.65 &  \textbf{~2.65} &  ~5.09 &  ~5.09 &  \textbf{~7.80} \\
\texttt{Transport-Sydney}     &  \textbf{~0.00} &  \textbf{~0.00} &  \textbf{~0.00} &  ~0.00 &  ~0.00 &  \textbf{~6.50} &  ~0.00 &  ~7.37 &  \textbf{~9.49} \\
\texttt{Web-NotreDame}        &  \textbf{~9.34} &  \textbf{~9.34} &  \textbf{~9.34} &  12.10 &  \textbf{13.79} &  \textbf{13.79} &  14.37 &  14.37 &  \textbf{19.22} \\
\bottomrule
\end{tabular}}

\caption{Average percentage eigen-drop on real networks (larger is better) when removing $p=1,10,$ or $100$ nodes. \texttt{Xdeg} is effective and has log-linear time in the number of nodes. Details about the sizes of these datasets are in Table~\ref{tab:real-data} of the appendix.}
\label{tab:imm-real}
\end{table*}

Table~\ref{tab:imm} shows the percentage reduction of the eigenvalue after immunization, averaged over repetitions on synthetic graphs. We can arrange immunization strategies in tiers according to increasing performance: strategies within a tier have comparable performance across data sets. The third tier is made up of \texttt{NS} and \texttt{degree}. They perform similarly because \texttt{NS} targets the largest eigenvalue of the adjacency matrix, which is largely dominated by node degree. Strategies in this tier perform substantially better than \texttt{core} (not shown), and are very close to the strategies in the next two tiers, i.e. \texttt{degree} is a very strong baseline in this task. The second tier is comprised of \texttt{CI} and \texttt{Xdeg}, with \texttt{Xdeg} having a slight advantage over \texttt{CI}. Finally, the best performance was achieved by \texttt{NB} and \texttt{XNB}. Their performances were almost indistinguishable in most data sets, though they have a small margin over \texttt{CI} and \texttt{Xdeg}.

Strategies in the best two tiers, i.e. \texttt{CI}, \texttt{Xdeg}, \texttt{NB} and \texttt{XNB}, all showed standard deviations of similar magnitude across all data sets (not shown), and the ordering in increasing performance \texttt{CI} < \texttt{Xdeg} < \texttt{NB} $\approx$ \texttt{XNB} is statistically significant at $p \ll 10^{-10}$ (see Appendix~\ref{app:immunization}). Further, the best two (\texttt{NB} and \texttt{XNB)} use the principal NB-eigenvector, whereas \texttt{CI} and \texttt{Xdeg} depend only on node degree, and are therefore much more computationally efficient.

Table~\ref{tab:imm-real} shows the results on real data sets, where we have run only \texttt{degree}, \texttt{CI}, and \texttt{Xdeg} for computational reasons. 
We use social networks \citep{de2013anatomy,LeskovecLDM09}, transportation networks, \citep{opsahlBlog,transportationRepo,LeskovecLDM09}, Autonomous Systems (AS) of the Internet networks \citep{ZhangLMZ04,karrer2014percolation}, and web crawl networks \citep{albert1999diameter}. See Section~\ref{app:data} for data set descriptions. We remove from each network $1$, $10$, and $100$ nodes at a time. Again, \texttt{degree} is a very strong baseline, but it is never better than both \texttt{CI} and \texttt{Xdeg} at the same time. All three strategies are able to drastically immunize the autonomous systems networks \texttt{AS-1} and \texttt{AS-2} at $100$ nodes removed, probably owing to the fact that their degree distribution is extremely heterogeneous and thus the nodes with largest degree have a large eigen-drop. In all other networks, \Xdeg{} achieves the best performance. An interesting case is that of \texttt{Transport-Sydney}. The node identified by all three strategies has an eigen-drop of exactly $0.0$. Following Corollary~\ref{cor:degree-one}, this means that the chosen node lies outside of the $2$-core of the graph and thus has no impact on non-zero NB-eigenvalues. After $10$ nodes are removed, both \texttt{degree} and \texttt{CI} continue to achieve zero eigen-drop, while \texttt{Xdeg} already identifies the correct nodes and ahieves $6.50\%$ decrease. Even at $100$ nodes removed, \texttt{degree} cannot identify nodes that generate an eigen-drop. A similar case occurs on \texttt{Transport-California}, where the first node identified by \texttt{degree} and \texttt{CI} generates no eigen-drop, while \texttt{Xdeg} is able to correctly identify influential nodes.

We conclude that in cases where efficiency is of the essence, \texttt{Xdeg} is the best overall immunization strategy, as it has a slight advantage over \texttt{CI} and its performance is close to optimal. If effectiveness is more important than efficiency, either \texttt{XNB} or \texttt{NB} should be used.


\section{Conclusion}\label{sec:conclusion}
We developed a theory of spectral analysis for the NB-matrix by studying what happens to its largest eigenvalue when one node is removed from the network. Our theory is independent of the structure of the graph, i.e. we make no assumptions of locally tree-like structure or density or length of cycles, as is usual in other studies. We find two new node-level statistics, or centrality measures, \Xnb{} centrality and \Xdeg{}, which are excellent predictors of a node's influence on the largest NB-eigenvalue. Finally, we focus on the application of targeted immunization, where we propose two new algorithms that are shown to be more effective than other strategies for a variety of real and synthetic graph ensembles.

Our techniques open many possibilities for further research. For instance, the left-hand side of Equation~(\ref{eqn:char-poly}) is reminiscent to certain quantities used in the theory of eigenvalue interlacing \citep{godsil2013algebraic}, while the matrix $\left( B' - tI \right)^{-1}$ on the right-hand side is known as the \emph{resolvent} of $B'$, which has many applications in random matrix theory \citep{taoStieltjes}. On a different note, \citet{cvetkovic1980spectra} highlight that most matrices associated to graphs are \emph{linear} combinations of $I$, $A$, and $D$, whereas the NB-matrix is associated with a \emph{quadratic} combination of $I$, $A$, and $D$, via Equation~(\ref{eqn:aux-nb-matrix}). In the future, we will explore which other matrices associated with graphs can be studied via quadratic, or higher order, combinations of $I$, $A$, and $D$.

We focused on the application to targeted immunization. However, other applications of NB-eigenvalues exist -- e.g., community detection and graph distance. Further studying the behavior of NB-eigenvalues under small perturbations of the graph, using the framework presented here, has potential to affect those applications.

\section*{Acknowledgements}
L.T. thanks Gabor Lippner for many invaluable discussions.

\bibliographystyle{unsrtnat}
\setcitestyle{numbers}
\bibliography{references}

\newpage
\appendix

\section{Technical Lemmas}\label{app:lemmas}
In this subsection, $B$ is the NB-matrix of a graph $G$, $P$ is defined in Section~\ref{sec:background}, and $\lambda, \mathbf{v}$ are the Perron eigenvalue and corresponding unit right eigenvector of $B$. Let also $\mathbf{v}^i = \sum_j \mathbf{v}_{j \to i}$ as in Equation~(\ref{eqn:nb-centrality}) and let $\overline{\mathbf{v}} = \left( \mathbf{v}^1, \ldots, \mathbf{v}^n \right)$.

\begin{lemma}\label{lem:right-left}
$P \mathbf{v}$ is a left eigenvector of $B$ corresponding to $\lambda$.
\end{lemma}

\begin{proof}
Since $PB$ is symmetric and $P^2 = I$, we have $B = P B^T P$. Now $B \mathbf{v} = \lambda \mathbf{v}$ implies $B^T P \mathbf{v} = \lambda P \mathbf{v}$, which completes the proof.
\end{proof}

\begin{lemma}\label{lem:small-big}
Suppose $\mathbf{v}$ is such that $\mathbf{v}^T P \mathbf{v} = 1$. Let $(\mathbf{f}, -\lambda \mathbf{f})$ be the left unit eigenvector of $B_{aux}$ corresponding to $\lambda$. Then, we have $\| \overline{\mathbf{v}} \| = \mu \| \mathbf{f} \|$, where
\begin{equation}
\mu = \sqrt{\frac{\lambda \left( \lambda^2 - 1 \right)} {1 - \mathbf{f}^T D \, \mathbf{f}}}.
\end{equation}
\end{lemma}

\begin{proof}
First, from $\mathbf{v}^T P \mathbf{v} = 1$ we get $\mathbf{v}^T P B \mathbf{v} = \lambda \mathbf{v}^T P \mathbf{v} = \lambda$, and we can expand $\mathbf{v}^T P B \mathbf{v}$ to find $\| \overline{\mathbf{v}} \|^2 - \| \mathbf{v} \|^2 = \lambda$. Second, since $\left(\mathbf{f}, -\lambda \mathbf{f} \right)$ has unit length, we have $\| \mathbf{f} \| ^2 = 1 / \left( \lambda^2 + 1 \right)$. Therefore,
\begin{equation}\label{eqn:v-bar-1}
\mu^2 = \left(\lambda^2 + 1 \right) \left( \lambda + \| \mathbf{v} \|^2 \right).
\end{equation}
Now, $B\mathbf{v} = \lambda \mathbf{v}$ implies $\mathbf{v}_{j \to i} + \lambda \mathbf{v}_{i \to j} = \mathbf{v}^i $ for any $i,j$. Plug this identity in $\| \mathbf{v} \|^2 = \sum_{i,j} a_{ij} \mathbf{v}_{i \to j}^2$ to find
\begin{equation}\label{eqn:v-bar-2}
\| \mathbf{v} \|^2 \left( \lambda^2 + 1 \right) + 2 \lambda = \sum_i \left( \mathbf{v}^i \right)^2 \deg i = \overline{\mathbf{v}}^T D \overline{\mathbf{v}} = \mu^2 \mathbf{f}^T D \, \mathbf{f}.
\end{equation}
Using (\ref{eqn:v-bar-1}) and (\ref{eqn:v-bar-2}) together finishes the proof.
\end{proof}

\begin{remark}
Both $\overline{\mathbf{v}}$ and $\mathbf{f}$ determine the same node centrality ranking, though the latter is easier to compute. However, the normalization $\mathbf{v}^T P \mathbf{v} = 1$ is fundamental in our theory, which makes $\overline{\mathbf{v}}$ the more appropriate choice. Lemma~\ref{lem:small-big} allows us to compute $\overline{\mathbf{v}}$ only with the knowledge of $\mathbf{f},\lambda$ and $D$, which is much more efficient than computing $\mathbf{v}$ and $\overline{\mathbf{v}}$ directly. 
\end{remark}

\section{Complexity Analysis of Algorithm~\ref{alg:xdeg}}\label{app:complexity}
In Section~\ref{sec:complexity-2} we used a standard map (i.e. hash table, or dictionary) to store the \Xdeg{} values in line $2$ of  Algorithm~\ref{alg:xdeg}. Alternatively, we can use an indexed priority queue (IPQ). An IPQ is a data structure that behaves like a priority queue except that, additionally, elements in the IPQ can be updated efficiently. The underlying data structure is a max-heap. An IPQ can find the maximum element in the heap, as well as update any element, in logarithmic time.

In this case, line $2$ of Algorithm~\ref{alg:xdeg} takes $m$ operations to compute the \Xdeg{} values plus $n$ operations to heapify the IPQ. Further, lines $4$ and $9$ take $O \left( \log n \right)$ time, which yields a time complexity of 
\begin{equation}
O \left( m + n + p \log n + \sum_{j=1}^p d_{i_j}+ \log n \sum_{j=1}^p D_{i_j} \right).
\end{equation}

\subsection{Homogeneous degree distribution}\label{app:complexity-homogeneous}
In networks with a homogeneous degree distribution (e.g. Poisson) we can estimate $d_{i_j} \approx \langle k \rangle$ and $D_{i_j} \approx \langle k \rangle^2 $, where $\langle k \rangle$ is the average degree. This yields $O \left( m + n + p \langle k \rangle^2 \log n \right)$ total complexity for the IPQ version, while the map version gives $O \left( m + p n + p \langle k \rangle^2 \right)$. If $p = O(n)$ and $\langle k \rangle = O(1)$, the IPQ version scales better in the worst case scenario.

\subsection{Heterogeneous degree distribution}\label{app:complexity-heterogeneous}
In networks whose degree distribution is well approximated by a power law, the probability of finding a node of degree $d$ scales as $d^{-\gamma}$, for some $\gamma > 0$. In this case, the first few nodes removed by Algorithm~\ref{alg:xdeg} will usually have large degree, comparable to the largest degree in the network, $d_{i_j} = O\left( d_{\max} \right)$ for each $j$. Further, in the worst case scenario, each of their neighbors will also have a degree comparable to $d_{\max}$ and thus $D_{i_j} = O \left( d_{\max}^2 \right)$ for each $j$. Using $d_{\max} = O \left( n^\frac{1}{\gamma - 1} \right)$ \citep{barabasi2014network} yields $O \left( m + p n + p n^\frac{2}{\gamma - 1} \right)$ for the map version and $O \left( m + p n^\frac{2}{\gamma - 1} \log n \right)$ for the IPQ version. In the typical case $ 2 \leq \gamma \leq 3$, the exponent $\frac{2}{\gamma - 1}$ varies between $1$ and $2$.

\subsection{Average runtime}\label{app:complexity-runtime}
We have provided the analysis of worst case scenario runtime. However, the average runtime of both the IPQ and map versions is close to linear, as shown in Figure~\ref{fig:scaling}. This figure was generated by first sampling a degree sequence from a power-law density $p_d \propto d^{-\gamma}$, and then generating a graph at random using the configuration model. Self-loops and multi-edges were removed and only the largest component was kept. Each marker is the average of $30$ repetitions. We used $p = 100$.

\begin{table*}
\resizebox{\textwidth}{!}{%
\begin{tabular}{r|ccrrrr}
\toprule
 & nodes & edges & $n$ & $m$ & $\lambda_1$ & $d_{\max}$ \\ 
\midrule
\texttt{AS-1} \citep{ZhangLMZ04} & AS & digital communication  &     34,761 &     107,720 &  151.442 &   2,760 \\
\texttt{AS-2} \citep{karrer2014percolation} & AS & digital communication &     22,963 &      48,436 &   64.678 &   2,390 \\
\texttt{Social-Slashdot} \citep{LeskovecLDM09} & users & friendships &     77,360 &     469,180 &  128.550 &   2,539 \\
\texttt{Social-Twitter} \citep{de2013anatomy} & users & friendships &    456,290 &  12,508,221 &  636.147 &  51,386 \\
\texttt{Transport-California} \citep{transportationRepo} & intersections & roads &  1,957,027 &   2,760,388 &    3.321 &      12 \\
\texttt{Transport-Sydney} \citep{transportationRepo} & intersections & roads     &     32,956 &      38,787 &    2.266 &      10 \\
\texttt{Web-NotreDame} \citep{albert1999diameter} & websites & hyperlinks      &    325,729 &   1,090,108 &  175.657 &  10,721 \\ 
\bottomrule
\end{tabular}}

\caption{Real-world data sets. $n$: number of nodes, $m$: number of edges, $\lambda_1$: largest NB-eigenvalue, $d_{\max}$: largest degree. AS stands for autonomous systems.}
\label{tab:real-data}
\end{table*}

\section{Experimental Setup}\label{app:experiments}

\subsection{Base lines}\label{app:baselines}

\paragraph{Degree.} The degree of a node $i$, denoted $d_i$ is the number of neighbors it has in the graph. Nodes of degree $1$ have zero Collective Influence, \Xdeg{}, NB-centrality, \Xnb{} centrality.

\paragraph{$k$-core index.} The $k$-core index of a node, also called \emph{coreness}, is defined as follows. First, iteratively remove all nodes of degree $1$ until there are none. All nodes removed in this step are assigned a value of $k$-core index of $1$. Then, iteratively remove all nodes of degree $2$; all nodes removed at this step have $k$-core $2$. Repeat this process until there are no more nodes in the graph. Notably, following Corollary~\ref{cor:degree-one}, all nodes with $k$-core value of $1$ have zero NB-centrality.

\paragraph{NB-centrality.} The NB-centrality of a node is defined in Equation~(\ref{eqn:nb-centrality}). It was proposed in \cite{radicchi2016leveraging} as an indicator of influential spreaders on locally tree-like networks for the SIR model.

\paragraph{NetShield.} NetShield is an efficient algorithm that identifies a subset of nodes with the highest ``shield-value'', which is defined as the impact a node, or set of nodes, has on the largest eigenvalue of the adjacency matrix \cite{ChenTPTEFC16}.

\paragraph{Collective Influence.}\label{app:ci}
The Collective Influence (CI) of node $i$ is
\begin{equation}
CI_i = \left( d_i - 1 \right) \sum_j a_{ij} \left( d_j -1 \right),
\end{equation}
though this definition can be generalized to include nodes in arbitrarily large neighborhoods around $i$ \cite{morone2015influence}. Note that this is quite similar in nature to \Xdeg{} in Equation~(\ref{eqn:x-degree-centrality-minus-one}). We think of \Xdeg{} as a second-order aggregation of the values $\left( d_j - 1 \right)$ of the neighbors of $i$, while CI is a first-order aggregation. Further, one can apply Algorithm~\ref{alg:xdeg} to perform targeted immunization based on CI instead of \Xdeg{}, and hence they have the same running time complexity (see Section~\ref{sec:complexity-2} and Appendix~\ref{app:complexity}). \citet{morone2016collective} claim that the CI algorithm runs in $O\left( n \log n \right)$ time, though we were not able to reproduce this result. In any case, any efficient algorithm that computes CI can be used to compute \Xdeg{} as well.

\subsection{Data sets}\label{app:data}
All synthetic graphs have $n=10^5$ nodes and parameters were chosen so that the average degree was approximately $12$. SBM graphs were generated with two blocks, or communities, so that the average within-block degree is $9$ and the between-block degree is $3$. 
WS graphs generated with rewiring probability $0.1$. BTER graphs were generated with target average local clustering coefficient of $0.98$, and target global clustering coefficient of $0.4$. BTER graphs were generated with the authors' implementation \citep{KoldaPPS14}; all other graphs were generated using NetworkX \citep{hagberg2008exploring} version 2.3. After generation, we extracted the largest connected component of each graph and converted all multi-edges to single edges and deleted self-loops. $100$ graphs were generated from each ensemble. Table~\ref{tab:real-data} describes the real data sets used. Directed networks were converted to undirected, and only the largest connected component of each data set was used.

\subsection{Approximating the Largest Eigenvalue}\label{app:approximating-the-eigenvalue}
Since nodes of large degree are bound to induce a larger eigen-drop than those of small degree, we chose target nodes at random by sampling $1\%$ of nodes from each graph, proportionally to their degree. This was achieved by sampling one edge at random, with replacement, and then choosing one of its endpoints randomly. This yields a probability of sampling node $i$ equal to $d_i / 2m$.

\subsection{Predicting the Eigen-drop}\label{app:predicting-the-eigengap}
Nodes were sampled in the same way as in  \ref{app:approximating-the-eigenvalue}. Figure \ref{fig:predicting} shows correlation coefficients, defined as the covariance divided by the product of the standard deviations of the two variables. We computed the correlation between the eigen-drop and each of the statistics: $\alpha$, $\widetilde{\alpha}$, \Xdeg{}, and degree. No three-way correlation was computed.

\subsection{Immunization with \mbox{X-NB} and \mbox{X-degree}}\label{app:immunization}
To confirm the ordering in increasing performance \texttt{CI} $<$ \texttt{Xdeg} $<$ \texttt{NB} $\approx$ \texttt{XNB}, we used a one-sided Wilcoxon signed-rank test, which is a non-parametric version of the paired T-test. In a paired sample setting, this test tests the null hypothesis that the median of the differences between the two samples is positive, against the alternative that it is negative. Therefore, a small $p$-value means that there is little probability that the first sample's median is smaller than the second's. For each graph ensemble and each percentage of removed nodes (1\%, 2\%, 3\%), the ranking \texttt{CI} $<$ \texttt{Xdeg} $<$ \texttt{NB} was confirmed with $p \ll 10^{-10}$ in all cases. Further, we have \texttt{NB} $<$ \texttt{XNB} in WS networks ($ p \ll 10^{-10}$) and BTER networks ($ p < 0.05$), and \texttt{NB} $>$ \texttt{XNB} in BA networks ($ p \ll 10^{-10}$) and SBM networks ($ p < 0.05$). We summarize these results by writing \texttt{CI} $<$ \texttt{Xdeg} $<$ \texttt{NB} $\approx$ \texttt{XNB}.

\end{document}